\documentclass[12pt]{article}
\usepackage{geometry}
\pdfoutput=1  % Ensure PDFlatex

\usepackage{amsmath}
\usepackage{amsthm}
\usepackage{amsfonts}
\usepackage{amssymb}
\usepackage{mathtools}

\usepackage{xr,hyperref}

\usepackage{natbib}

 \usepackage{parskip}
 	    \makeatletter
 		\def\thm@space@setup{%
 		  \thm@preskip=\parskip \thm@postskip=0pt
 		}
 		\makeatother

\newcommand{\blind}{0}

\usepackage{authblk}

\usepackage{xcolor,calc}

\newcommand{\tpose}{^\top}
\DeclareMathOperator{\tr}{tr}

\DeclarePairedDelimiter\floor{\lfloor}{\rfloor}
\DeclareMathOperator{\Prob}{\mathbb{P}}
\newcommand{\norm}[1]{\left\lVert#1\right\rVert}

\newtheorem{theorem}{Theorem}[section]

\newtheorem{lemma}{Lemma}[section]

\newtheorem{definition}{Definition}[section]

\newtheorem*{remark}{Remark}

\newcommand{\vect}[1]{\boldsymbol{#1}}
\newcommand{\mat}[1]{\boldsymbol{#1}}

\newcommand{\mean}{\vect{\mu}}
\newcommand{\var}{\mat{\Sigma}}
\newcommand{\meanS}[1]{\vect{\mu}_{#1}} %sample mean
\newcommand{\varS}[1]{\hat{\mat{\Sigma}}_{#1}} %sample variance
\newcommand{\varU}[1]{\mat{\Sigma}_{#1}} %unbiased sample variance
\newcommand{\sample}[1]{\vect{\xi}^{(#1)}}
\newcommand{\randvar}{{\vect{\xi}}}
\newcommand{\etavar}{{\vect{\eta}}}

\newcommand{\randvarUni}{\xi}
\newcommand{\sampleUni}[1]{\xi^{(#1)}}
\newcommand{\meanUni}{\mu}

\newcommand{\meanSUni}[1]{\mu_{#1}} %sample mean
\newcommand{\varUUni}[1]{\sigma_{#1}} %unbiased sample mean

\usepackage{color} % Change color of text
\usepackage{graphicx}
\usepackage{todonotes}

\begin{document}

\if0\blind
{
% Document details
\title{\bf Multivariate Chebyshev Inequality with Estimated Mean and Variance}
%\author{Bartolomeo Stellato,  \quad Bart P.\ G.\ Van Parys, \quad and \quad Paul J.\ Goulart}
\author[1]{Bartolomeo Stellato}
\author[2]{Bart P.\ G.\ Van Parys}
\author[1]{Paul J.\ Goulart}
\affil[1]{Department of Engineering Science, University of Oxford}
\affil[2]{Operations Research Center, Massachusetts Institute of Technology}
\date{}
  \maketitle
% AUTHORS FOOTNOTE
\let\thefootnote\relax\footnote{Bartolomeo Stellato (e-mail: \textit{bartolomeo.stellato@eng.ox.ac.uk}) and Paul J.\ Goulart (e-mail: \textit{paul.goulart@eng.ox.ac.uk}), Control Group, Department of Engineering Science, University of Oxford, Parks Road, Oxford OX1 3PJ, United Kingdom; Bart P.\ G.\ Van Parys (e-mail: \textit{vanparys@mit.edu}) Operations Research Center, Massachusetts Institute of Technology, Cambridge, Massachusetts 02139. The authors would like to thank Sergio Grammatico for the fruitful suggestions.
}
} \fi

\if1\blind
{
  \bigskip
  \bigskip
  \bigskip
  \begin{center}
    {\LARGE\bf Multivariate Chebyshev Inequality with Estimated Mean and Variance}
\end{center}
  \medskip
} \fi

\bigskip
\begin{abstract}
A variant of the well-known Chebyshev inequality for scalar random variables can be formulated in the case where the mean and variance are estimated from samples. In this paper we present a generalization of this result to multiple dimensions where the only requirement is that the samples are independent and identically distributed. Furthermore, we show that as the number of samples tends to infinity our inequality converges to the theoretical multi-dimensional Chebyshev bound.
\end{abstract}

\noindent%
\vfill

% \newpage
% \spacingset{1.45} % DON'T change the spacing!

\section{Introduction}
The Chebyshev inequality \citeyearpar{Chebyshev:1867vg} is a fundamental result from probability theory and has been studied extensively  for more than a century in a wide range of sciences. The most common version of this result asserts that the probability that a scalar random variable $\randvarUni$ with distribution $\Prob$ differs from its mean $\meanUni \in \mathbb{R}$ by more than $\lambda\in \mathbb{R}_{>0}$ standard deviations $\sigma \in \mathbb{R}_{>0}$ satisfies the relation
\begin{equation}\label{eq:chebyshev_inequality}
\Prob \left(|\xi - \mu| \geq \lambda \sigma\right)\leq
 \min\left\{1,\frac{1}{\lambda^2}\right\}.
\end{equation}

Recent works by \citet{Chen:2007wt} and \citet{Navarro:fm} provide a closed form extension of \eqref{eq:chebyshev_inequality} to the multivariate case where the confidence intervals are ellipsoids centered at the population mean. Moreover, \citet{Navarro:2014ka} shows that the derived extension is tight. Another extension  of \eqref{eq:chebyshev_inequality} to more general ellipsoidal and polyhedral sets has been described in \citep{Vandenberghe:2007iv} where a multivariate version of the Chebyshev bound is computed as a solution to a semidefinite program (SDP)~\citep{Vandenberghe:1996fy}.

Although these results provide means to explicitly compute distribution-free probability bounds based only on the first two moments of $\Prob$, they are of limited practical value since one often does not know $\mu$ and $\sigma$ exactly. In practice, a common approach is to compute empirical estimates of $\mu$ and $\sigma$ via sampling and to substitute these estimated values into \eqref{eq:chebyshev_inequality}, although it can be shown that this method can lead to unreliable results in the event of poor estimates of the moments.
There is an extensive literature about empirical processes where the quality of the estimates is investigated, see e.g.\ \citet{Dudley:1978tl} and \citet{vandeGeer:2010tk}. However, these approaches suffer from two main problems. They either assume that the underlying distribution $\Prob$ has bounded support (e.g.\ Hoeffding's inequality), or they provide asymptotic results on the convergence rate that are valid only as the sample size tends to infinity. Unfortunately, neither of these two cases turns out to be helpful when we make no assumption on the support of the distribution and the number of samples is limited.

In the univariate case, \citet{Saw:1984tz} approached the problem of formulating an empirical Chebyshev inequality from a different direction.  Given $N$ i.i.d samples $\sampleUni{i},\dots,\sampleUni{N} \in \mathbb{R}$ from an unknown distribution $\Prob$, and their empirical mean $\meanSUni{N}$ and empirical standard deviation (SD) $\varUUni{N}$, Saw derives a Chebyshev inequality with respect to the $(N+1)^{\rm th}$ sample. The bound derived is remarkably simple and requires only a modification of the right-hand side of the theoretical bound in~\eqref{eq:chebyshev_inequality}, i.e.\
%\PJGnote{Is the $\sigma_N$ here biased or unbiased?  I have assumed unbiased}
\begin{equation}\label{eq:sampled_chebyshev}
\begin{multlined}[0.4\textwidth]
\Prob^{N+1}\left(\left|\sampleUni{N+1}- \meanSUni{N}\right| \geq \lambda \varUUni{N}\right) \\
\leq \min\left\{1,\frac{1}{N+1} \floor*{\frac{(N+1)(N^2 - 1 + N\lambda^2)}{N^2\lambda^2}}\right\},
\end{multlined}
\end{equation}
where $\floor{\cdot}$ denotes the floor function\footnote{$\lambda$ corresponds to $k\sqrt{\frac{N+1}{N}}$ in~\citep{Saw:1984tz}. }.
 % \BSnote{Reviewer 2 asked to specify: $\lambda = k \sqrt{\frac{N+1}{N}}$. We can specify what $\lambda$ is to link this definition to Saw one to make it clear, but in my opinion, this creates only more confusion.}

% For simplicity, we simplify the right-hand side in~\eqref{eq:sampled_chebyshev1} to
% \begin{equation}\label{eq:sampled_chebyshev}
% \begin{multlined}[0.4\textwidth]
% \Prob^{N+1}\left(\left|\sampleUni{N+1}- \meanSUni{N}\right| > \lambda \varUUni{N}\right) \\
% \leq \min\left\{1,\frac{(N^2 - 1 + N\lambda^2)}{N^2\lambda^2}\right\}.
% \end{multlined}
% \end{equation}
Currently, there exists no counterpart of \eqref{eq:sampled_chebyshev} for the multivariate case. There have been only limited efforts to extend these results to multiple dimensions by making strong assumptions on the population. In \citep{Navarro:2014uk} the author derives a multivariate equivalent by assuming that the true distribution of the population is the empirical distribution over a given data set.

%
% include the sampled mean and covariance inside a multivariate extension of the Chebyshev inequality by making strong assumptions on the population, e.g.\ in \citep{Navarro:2014uk} the author assumes that, given a data set, the population has an empirical distribution over the data, without taking in account any new possible data point.

In this paper we derive a multivariate version of the inequality in~\eqref{eq:sampled_chebyshev} using the Euclidean norm, without requiring any further assumptions on the distrubution.
% will show that the extension of this inequality to the multivariate case using the Euclidean norm is straightforward without any further assumptions on the distribution and requires only the multiplication of the right-hand side by the dimension of the random vector.
In addition, we show that the result converges to the multivariate Chebyshev inequality as computed in \citep{Vandenberghe:2007iv} for an ellipsoidal set centered at the mean.

\section{Main Results}
% \BSnote{The reviewer 2 asked to put all vectors and matrices in bold characters.}
Before stating the main result, we require the following definition: \\

\begin{definition}
Let $\randvar\in \mathbb{R}^{n_{\randvar}}$ be a random variable and let $N\in \mathbb{Z}_{\geq n_\randvar}$.
%Let $(\Omega, \mathcal{F},\Prob)$ be  a probability space where: $\Omega$ is a space of elementary events, $\mathcal{F}$ is a $\sigma$-algebra of subsets of $\Omega$, $\Prob$ is a probability measure defined on the events in $\mathcal{F}$. Moreover, let the mapping $X: \Omega \to \mathbb{R}^{d}, \omega \mapsto x$, be a random variable.
Given $(N+1)$ i.i.d.\ samples $\sample{1}, \dots, \sample{N}, \sample{N+1}\in \mathbb{R}^{n_{\randvar}}$ of $\randvar$ with mean $\mu\in \mathbb{R}^{n_{\randvar}}$ and covariance matrix $\var\in \mathbb{R}^{n_{\randvar}\times n_{\randvar}}$, we define the \emph{empirical mean} as
\begin{equation}
	\meanS{N} \coloneqq \frac{1}{N}\sum_{i=1}^{N}\sample{i},
\end{equation}
and the \emph{unbiased} and \emph{biased empirical covariances} as
\begin{equation*}
\varU{N} \coloneqq \frac{1}{N-1} \sum_{i=1}^{N}(\sample{i} - \meanS{N})(\sample{i} - \meanS{N})\tpose, \qquad \varS{N} \coloneqq \frac{1}{N} \sum_{i=1}^{N}(\sample{i} - \meanS{N})(\sample{i} - \meanS{N})\tpose,
\end{equation*}
respectively.
\end{definition}

We can now state our main result, which is a multivariate version of the univariate result of \citet{Saw:1984tz}: \\

\begin{theorem}\label{thm:mutiv_cheb_ineq}
Let $\randvar\in \mathbb{R}^{n_{\randvar}}$ be a random variable and let $N\in \mathbb{Z}_{\geq n_\randvar}$.
%Let $(\Omega, \mathcal{F},\Prob)$ be  a probability space where: $\Omega$ is a space of elementary events, $\mathcal{F}$ is a $\sigma$-algebra of subsets of $\Omega$, $\Prob$ is a probability measure defined on the events in $\mathcal{F}$. Moreover, let the mapping $X: \Omega \to \mathbb{R}^{d}, \omega \mapsto x$, be a random variable.
Given $N+1$ i.i.d\ samples of $\randvar$ denoted as $\sample{1}, \dots, \sample{N}, \sample{N+1}\in \mathbb{R}^{N_{\randvar}}$, if we assume that $\varU{N}$ is nonsingular,
then for all $\lambda \in \mathbb{R}_{>0}$ it holds that:
\begin{equation}\label{eq:multivariate_sampled_chebyshev}
\begin{multlined}[0.4\textwidth]
\Prob^{N+1}\left((\sample{N+1}- \meanS{N})\tpose\varU{N}^{-1}(\sample{N+1} - \meanS{N}) \geq \lambda^2 \right) \\
\leq \min\left\{1,\frac{1}{N+1}\floor*{\frac{n_{\randvar} (N+1) (N^2 - 1 + N\lambda^2)}{N^2\lambda^2}}\right\}.
\end{multlined}
\end{equation}
%where $\varU{N}^{\dagger}$ is the Moore-Penrose pseudoinverse\footnote{$\varU{N}^{\dagger} = (\varU{N}\tpose \varU{N})^{-1}\varU{N}^{\top}$} of $A$ and
\end{theorem}

\begin{remark}
  The inequality~\eqref{eq:multivariate_sampled_chebyshev} can be simplified by upper bounding the floor function by its argument
\begin{equation*}
  \begin{multlined}[0.4\textwidth]
  \Prob^{N+1}\left((\sample{N+1}- \meanS{N})\tpose\varU{N}^{-1}(\sample{N+1} - \meanS{N}) \geq \lambda^2 \right) \\
  \leq \min\left\{1,\frac{n_{\randvar} (N^2 - 1 + N\lambda^2)}{N^2\lambda^2}\right\}.
  \end{multlined}
\end{equation*}
\end{remark}
We can also show that our empirical bound is well behaved in the limit as $N\to\infty$, coinciding with the (tight) analytical bound computable using the method of \citet{Vandenberghe:2007iv}: \\

\begin{theorem}\label{thm:limit}
As $N\to\infty$, the right-hand side of \eqref{eq:multivariate_sampled_chebyshev} tends to
\begin{equation*}
\min \left\{1, \frac{n_{\randvar}}{\lambda^2} \right\},
\end{equation*}
%\PJGnote{This appears to be the same as the result in the paper by Chen (2011) I found on Arxiv (see directory).  Note sure if the LMI based proof is necessary. }
which corresponds to the Multivariate Chebyshev inequality over ellipsoids shaped according to $\var$ and centered in $\mean$.
\end{theorem}

%\begin{remark}
%%The assumption that $\varU{N}$ is nonsingular can be justified in practice.
%In the case when $\randvar$ is a continuous random variable, the true covariance $\var \in \mathbb{S}^{n_{\randvar}}_{+}$ is nonsingular and the number of samples is greater than $n_{\randvar}$, it can be shown that the probability of having a singular $\varU{N}$ is null.
%\end{remark}

\section{Proof of the Main Results}

In order to prove our main results, we require two supporting lemmas. \\

%\PJGnote{Come back to notation for this result}
\begin{lemma}\label{lem:lemA}
%\todo[inline]{Use of strict/non-strict inequalities is not consistent here.  Fix either by defining $J$ in terms of $\norm{u_i} \ge k$, or strengthen the bound on J to a strict inequality (either works).  Proof could be more compact.  Why is $N \ge 2$ required?}
Let $k\in \mathbb{R}_{>0}$ and $N \in \mathbb{Z}_{\geq n_\randvar}$. Consider a set of vectors $\mathcal{U}_N \coloneqq \left\{\vect{u}_i\right\}_{i=1}^{N}$ with $\vect{u}_i \in \mathbb{R}^{n_{\randvar}}$ for all $i \in \{1,\dots,N\}$ satisfying the conditions
\begin{equation*}
\sum_{i=1}^{N} \vect{u}_{i} = 0_{n_{\randvar}}, \qquad \sum_{i=1}^{N} \vect{u}_{i}\vect{u}_{i}^{\top} = N \mat{I}_{n_{\randvar}\times n_{\randvar}}.
\end{equation*}
Define the subset of vectors in $\mathcal{U}_N$ with norm greater or equal to $k$ as
\begin{equation*}
J(\mathcal{U}_N,k) \coloneqq \left\{\vect{u}_i \in \mathcal{U}_N\;:\;\norm{\vect{u}_{i}}_{2} \geq k\right\},
\end{equation*}
where $\|\cdot\|_2$ is the Euclidean norm.
%number of such vectors with norm larger than $k$ has
%\begin{equation*}
%J(\mathcal{U}_N,k) \coloneqq \left|\left\{u_i \in \mathcal{U}_N\;:\;\norm{u_{i}}_{2}> k\right\}\right|.
%\end{equation*}
Then the cardinality of $J(\mathcal{U}_N,k)$ is bounded by
$%\begin{equation}\label{eq:multiv_cheb_lemma1}
{\left|J(\mathcal{U}_N,k)\right|  \leq \floor*{\frac{n_{\randvar} N}{k^2}}
}$.%\end{equation}
\end{lemma}
\begin{proof}
Observe that
\begin{equation*}
\norm{\vect{u}_{i}}_{2} \geq k \quad \iff \quad \vect{u}_i\tpose \vect{u}_i \geq k^2.
\end{equation*}
Summing both sides of the preceding inequality over $J(\mathcal{U}_N, k)$ produces
\begin{equation*}
k^2 \left|J(\mathcal{U}_N,k)\right| \leq \sum_{\vect{u}_i \in J(\mathcal{U}_N, k)}\vect{u}_i\tpose \vect{u}_i \leq \sum_{i=1}^{N}\vect{u}_i\tpose \vect{u}_i = \tr\left(\sum_{i=1}^{N}\vect{u}_i\vect{u}_i\tpose\right) = n_{\randvar} N
\end{equation*}
and the result follows immediately.
\end{proof}
~\\

%\begin{lemma}\label{lem:multiv_cheb2}
%Given a fixed set $U \coloneqq \left\{u_1, \dots,u_{n}\right\}$ with $u_i \in \mathbb{R}^{d}$ for all $i \in \{1,\dots,n\}$ such that
%\begin{equation}
%\sum_{i=1}^{n} u_{i} = 0_{d}, \qquad \sum_{i=1}^{n} u_{i}u_{i}^{\top} = n I_{d\times d},
%\end{equation}
%then, it is possible to define a distribution $\Prob_{U}$ such that, for all $t \in \mathbb{R}^{d}$:
%\begin{equation}
%\Prob_{U}\left(\left\{y \in U \;:\;  y = t\right\}\right) = \frac{\left\vert \{i\in\{1,\dots, n\}\;:\; u_{i} = t\}\right\vert}{n},
%\end{equation}
%%and that
%%\begin{equation}
%%\Exp\{y\}= 0_{d}, \qquad \Exp\{yy\tpose\} = I_{d\times d}.
%%\end{equation}
%%Thus,
%and
%\begin{equation}
%\Prob_{U}\left(\left\{y\in U \;:\;  \norm{y}_2 > k\right\}\right)\leq \frac{1}{n} \floor*{\frac{d n}{k^2}}.
%\end{equation}
%\end{lemma}
%
%\begin{proof}
%From the definition of $y$ it is possible to write:
%\begin{equation*}
%\Prob_{U}\left(\left\{ y\in U \;:\;  \norm{y}_2 > k \right\}\right)\leq \frac{1}{n} \max \left\vert \left\{i\in\{1,\dots, n\}\;:\; \norm{u_{i}}_2 > t\right\}\right\vert \leq \frac{1}{n} \floor*{\frac{d n}{k^2}},
%\end{equation*}
%where last inequality comes from Lemma~\ref{lem:multiv_cheb1}.
%\end{proof}

\begin{lemma}\label{lem:lemB}
The following relations hold:
\begin{align}
&\sample{N+1} - \meanS{N+1} = \frac{N}{N+1}(\sample{N+1} - \meanS{N}) \label{prop:multiv_cheb1}\\
&\varS{N+1} = \frac{N-1}{N+1}\varU{N} + \frac{N}{(N+1)^2}(\sample{N+1} - \meanS{N})(\sample{N+1} - \meanS{N})\tpose.\label{prop:multiv_cheb2}
\end{align}
\end{lemma}

\begin{proof}
The first relation can be obtained directly by writing
\begin{equation*}
\begin{aligned}
\sample{N+1} - \meanS{N+1} &= \sample{N+1} - \frac{1}{N+1}\left(\sample{N+1} + \sum_{i=1}^{N}\sample{i}\right)%
%&= \frac{N}{N+1}\sample{N+1} - \frac{1}{N+1}\sum_{i=1}^{N}\sample{i}\\
%&= \frac{N}{N+1}\left(\sample{N+1} - \meanS{N}\right).
\end{aligned}
\end{equation*}
and collecting terms; the same result appears in \citep{Welford:1962tw}.  The second relation can be found by first defining the partial sums $\mat{S}_N$ as
\[
\mat{S}_N \coloneqq \sum_{i=1}^{N}\left(\sample{i} -\meanS{N}\right)\left(\sample{i} -\meanS{N}\right)\tpose,
\]
for which \citet{Welford:1962tw} provides  (with obvious modifications) the recurrence relation
\[
\mat{S}_{N+1} = \mat{S}_{N} + \frac{N}{N+1}\left(\sample{N+1}-\meanS{N}\right)\left(\sample{N+1}-\meanS{N}\right)\tpose.
\]
The result then follows by applying the identities $(N+1)\varS{N+1} = \mat{S}_{N+1}$ and $(N-1)\varU{N} = \mat{S}_N$.
\end{proof}

We are now in a position to prove both of our main results:

\subsection*{Proof of Theorem~\ref{thm:mutiv_cheb_ineq}}

 Since $\varU{N}$ is assumed nonsingular, it follows that $\varS{N+1} \succeq \varU{N} \succ 0$, i.e.\ $\varS{N+1}$ is positive definite.

 % is positive semidefinite by construction, i.e. $\varU{N}\succeq 0$, and since we assumed that $\varU{N}$ is nonsingular, we know that $\varU{N}\succ 0$, i.e. it is positive definite.

  %we assumed that $\varU{N}$ is nonsingular and from the construc%\footnote{We notice that both $\meanS{N+1}$ and $\varU{N}$ are positive semidefinite by construction.}
% $\varU{N}\succ 0$,

% It follows from \eqref{prop:multiv_cheb2} that also $\varS{N+1} \succ 0$.

%\textcolor{blue}{[TODO: In Saw paper they separate the case (in one dimension!) when $\meanS{N+1}=0$  from the one when $\meanS{N+1}>0$. (Paragraph after equation (3.3)). Then they show that when $\meanS{N+1}=0$ the Probability measure of the inequality is always $0$. See if it make sense to work out this case or not.]}
%
%
%\textcolor{blue}{[REMARK: We can state as an assumption that $\meanS{N+1}$ and $\varU{N}$ are not singular.]}

%\PJGnote{Possibly the full rank assumption on $\Sigma$ can be removed by assuming here that all data can be normalized into some space with $I$ replaced with some partial diagonal $I$.}
%
Normalize each of the $N+1$ samples $\sample{i}$ using
\begin{equation}\label{eq:multiv_cheb_ui}
\vect{u}_i \coloneqq \left(\varS{N+1}\right)^{-1/2}\left(\sample{i} - \meanS{N+1}\right), \quad \forall i \in \left\{1,\dots,N+1\right\}
\end{equation}
so that
\begin{equation}\label{eq:cond_ui_lemma}
\sum_{i=1}^{N+1} \vect{u}_{i} = 0_{d}, \quad \sum_{i=1}^{N+1} \vect{u}_{i}\vect{u}_{i}^{\top} = (N+1) \mat{I}_{n_{\randvar}\times n_{\randvar}},
\end{equation}
and~\eqref{eq:cond_ui_lemma} satisfies Lemma~\ref{lem:lemA}. Since all of the vectors $\vect{u}_i$ are i.i.d. and not more than $J(\mathcal{U}_{N+1},k)$ of these $N+1$ vectors have norm greater or equal to $k$, we have from Lemma~\ref{lem:lemA} that
\begin{equation}\label{eqn:Prob_ui_form}
\Prob^{N+1} \left(\norm{\vect{u}_{N+1}}_2 \geq k\right)
\le \frac{J(\mathcal{U}_{N+1},k)}{N+1} = \frac{1}{N+1} \floor*{\frac{n_{\randvar} (N+1)}{k^2}}.
\end{equation}

Considering next the inequality $\norm{\vect{u}_{N+1}}_2 \geq k$,  apply \eqref{prop:multiv_cheb1} and \eqref{eq:multiv_cheb_ui} to obtain the equivalent condition
\[
\left(\sample{i} - \meanS{N+1}\right)\tpose\varS{N+1}^{-1}\left(\sample{i} - \meanS{N+1}\right) \geq k^2.
\]

By using Lemma~\ref{lem:lemB} it is possible to define $\varS{N+1}$ as%
\begin{equation}\label{eq:recurs_N_plus_1}
  \varS{N+1} = \frac{N-1}{N+1}\varU{N} + \frac{1}{N}\left(\sample{i} - \meanS{N+1}\right)\left(\sample{i} - \meanS{N+1}\right)^\top.
\end{equation}

Let us define $q_N \coloneqq \left(\sample{i} - \meanS{N}\right)\tpose\varU{N}^{-1}\left(\sample{i} - \meanS{N}\right)$ and $q_{N+1}\coloneqq \left(\sample{i} - \meanS{N+1}\right)\tpose\varU{N}^{-1}\left(\sample{i} - \meanS{N+1}\right)$. Note that $q_{N+1} = \frac{N^2}{(N+1)^2}q_{N}$ from~\eqref{prop:multiv_cheb1}. Applying the Sherman-Woodbury-Morrison identity~\citep{sherman1949adjustment} to~\eqref{eq:recurs_N_plus_1} we can invert matrix $\varS{N+1}$ obtaining
\begin{align*}
\left(\sample{i} - \meanS{N+1}\right)\tpose\varS{N+1}^{-1}\left(\sample{i} - \meanS{N+1}\right) &=
 \frac{N+1}{N-1}q_{N+1} -
\left(1 + \frac{1}{N}\frac{N+1}{N-1}q_{N+1}\right)^{-1}\frac{1}{N}\left(\frac{N+1}{N-1}\right)^2q_{N+1}^2 \\[1ex]
&= \frac{(N+1)q_{N+1}}{(N-1)N + (N+1)q_{N+1}} \geq k^2.
\intertext{where the first equality has been pre- and post-multiplied by $\left(\sample{i} - \meanS{N+1}\right)^\top$ and $\left(\sample{i} - \meanS{N+1}\right)$ respectively. The latter inequality can be rewritten in terms of $q_N$ and rearranged to}
q_{N} &\geq \frac{(N^2-1)k^2}{N(N-k^2)},
\end{align*}

% Defining $q\coloneqq \left(\sample{i} - \meanS{N+1}\right)\tpose\varU{N}^{-1}\left(\sample{i} - \meanS{N+1}\right)$ and inverting the matrix $\varS{N+1}$ by applying the Sherman-Woodbury-Morrison identity~\citep{sherman1949adjustment} to \eqref{prop:multiv_cheb2} produces
% %
% \begin{align*}
% \left(\sample{i} - \meanS{N+1}\right)\tpose\varS{N+1}^{-1}\left(\sample{i} - \meanS{N+1}\right) &=
%  \frac{N+1}{N-1}q - \frac{N}{N^2-1}
% \left(1 + \frac{N}{N^2-1}q\right)^{-1}q^2 \\[1ex]
% &= \frac{N^2q}{(N^2-1) + Nq}
% %
% > k^2.
% \intertext{
% The latter inequality can be rearranged to}
% q &> \frac{(N^2-1)k^2}{N(N-k^2)},
% \end{align*}
% %
so that \eqref{eqn:Prob_ui_form} is equivalent to
\begin{equation}\label{eqn:Prob_k_form}
\Prob^{N+1} \left(\left(\sample{i} - \meanS{N}\right)\tpose\varU{N}^{-1}\left(\sample{i} - \meanS{N}\right) \geq \frac{(N^2-1)k^2}{N(N-k^2)}\right)
\le \frac{1}{N+1} \floor*{\frac{n_{\randvar} (N+1)}{k^2}}.
\end{equation}

Finally, define $\lambda$ such that
\begin{equation*}
\lambda^2 = \frac{(N^2 - 1)k^2}{N(N-k^2)}, \quad \text{so that} \quad k^2 = \frac{N^2\lambda^2}{N^2 - 1 + N\lambda^2}.
\end{equation*}
%\PJGnote{Are negative values problematic here?}
Direct substitution into \eqref{eqn:Prob_k_form} then produces the desired inequality
\begin{equation}
\begin{aligned}
\Prob^{N+1}&\left((\sample{N+1}- \meanS{N})\tpose\varU{N}^{-1}(\sample{N+1} - \meanS{N}) \geq \lambda^2\right)
% &\leq \frac{1}{N+1}\floor*{\frac{n_{\randvar} (N+1)(N^2 - 1 + N\lambda^2)}{N^2\lambda^2}} \\
\leq \frac{1}{N+1}\floor*{\frac{n_{\randvar} (N+1)(N^2 - 1 + N\lambda^2)}{N^2\lambda^2}}.
\end{aligned}\tag*{\qed}
\end{equation}%
%\begin{flushright}
%$\blacksquare$
%\end{flushright}

\subsection*{Proof of Theorem~\ref{thm:limit}}
Given $\mean\in \mathbb{R}^{n_{\randvar}}$ and $\var \in \mathbb{R}^{n_{\randvar} \times n_{\randvar}}$, $\var \succeq 0$, %$\Sigma \in \mathbb{S}^{n_{\randvar}}_{+}$
as the mean and covariance of the random variable $\randvar\in \mathbb{R}^{n_{\randvar}}$  respectively, we now derive the multivariate Chebyshev inequality bounding the probability
\begin{equation}\label{eq:multiv_cheb_proof_1}
\Prob \left((\randvar - \mean)\tpose \var^{-1}(\randvar-\mean) \geq \lambda^2\right),
\end{equation}
which is the probability of the complement of the ellipsoid shaped by $\var$ centered at the mean $\mean$.

% over ellipsoids shaped as $\var^{-1}$ and centered in $\mean$ corresponds to the upper bound of
% \begin{equation}\label{eq:multiv_cheb_proof_1}
% \Prob \left((\randvar - \mean)\tpose \var^{-1}(\randvar-\mean)> \lambda^2\right).
% \end{equation}
Without loss of generality, we shift the coordinate system to the mean $\mean$ by defining the variable $\etavar \coloneqq \randvar - \mean$ with zero mean $\mean_{\etavar} =0$ and variance $\var_{\eta} = \var$. Let us define $\mathcal{E}$ as the ellipsoid
\begin{equation*}
\mathcal{E} \coloneqq \left\{\etavar\tpose \frac{\var^{-1}}{\lambda^2}\etavar - 1 < 0\right\}.
\end{equation*}
The problem of computing an upper bound on the probability of $\etavar$ falling in the complement $\mathcal{E}^c$ of the ellipsoid $\mathcal{E}$ is equivalent to bounding the probability \eqref{eq:multiv_cheb_proof_1}. Let $1_{\mathcal{E}^c}(\cdot)$ denote the indicator function of set $\mathcal{E}^c$, i.e. if $1_{\mathcal{E}^c}(\etavar) = 0$ if $\etavar \notin \mathcal{E}^c$ and $1_{\mathcal{E}^c}(\etavar ) = 1$ if $\etavar  \in \mathcal{E}^c$; with the obvious relation $\Prob\left(\etavar  \in \mathcal{E}^{c}\right) = \mathbb{E}\left(1_{\mathcal{E}^{c}}\right)$. In order to bound the latter, we can define a quadratic function \citep[Section 7.4.1]{Boyd:2004uz} $f(\etavar ) = \etavar ^\top \mat{P} \etavar  + 2\vect{q}^\top \etavar  + r$ such that $f(\etavar ) \geq 1_{\mathcal{E}^c}(\etavar )$ for all $\etavar  \in \mathbb{R}^{n_\etavar }$. Equivalently, this inequality can be written as $f(\etavar ) \geq 1, \forall \etavar  \in \mathcal{E}^c$ and  $f(\etavar ) \geq 0, \forall \etavar  \in \mathbb{R}^{n_{\etavar }}$. By taking the expected value we obtain
\begin{equation}
  \mathbb{E}(f(\etavar )) \geq \mathbb{E}\left(1_{\mathcal{E}^{c}}(\etavar )\right) = \Prob\left(\etavar  \in \mathcal{E}^{c}\right).
\end{equation}
Hence, the problem of upper-bounding~\eqref{eq:multiv_cheb_proof_1} is equivalent to solving the convex problem
\begin{subequations}\label{eq:sdp_upper_bound}
 \begin{alignat}{2}
&\text{minimize} \quad & & \mathbb{E}(f(\etavar )) \\
           & \text{subject to:} %           & &  P \in \mathbb{S}^{n_{\randvar}},\; q \in \mathbb{R}^{n_{\randvar}},\; r \in \mathbb{R}, \; \tau \in \mathbb{R}\\
           && f(\etavar ) \geq 1,\; \etavar  \in \mathcal{E}^{c} \label{eq:sdp_upper_bound:constr1}\\
&&& f(\etavar ) \geq 0, \; \etavar  \in \mathbb{R}^{n_\etavar }.\label{eq:sdp_upper_bound:constr2}
\end{alignat}
\end{subequations}
Since $f(\etavar )$ is a quadratic function and  we know the first and second moments of $\etavar $, we can compute its expected value as
\begin{equation}\label{eq:sdp:exp_value}
  \mathbb{E}(f(\etavar )) = \tr\left(\Sigma_{\etavar}  \mat{P}\right) + 2\vect{q}^\top \mean_{\etavar}  + r =  \tr\left(\Sigma \mat{P}\right) + r,
\end{equation}
where $\mean_{\etavar}  = 0$ and $\var_{\etavar}  = \var$. The constraint~\eqref{eq:sdp_upper_bound:constr2} can be rewritten as the following linear matrix inequality (LMI)
\begin{equation}\label{eq:sdp:cond1}
  \begin{bmatrix}
      \mat{P} & \vect{q}\\
      \vect{q}\tpose & r
      \end{bmatrix} \succeq 0,
\end{equation}
see~\citep{Vandenberghe:1996fy}. By making use of the S-procedure \citep[Section B.2]{Boyd:2004uz}, we can define a scalar $\tau\geq 0$ and rewrite~\eqref{eq:sdp_upper_bound:constr1} as another LMI
\begin{equation} \label{eq:sdp:cond2}
  \begin{bmatrix}
  \mat{P} & \vect{q}\\
  \vect{q}\tpose & r - 1
  \end{bmatrix} \succeq \tau  \begin{bmatrix}
  \frac{\var^{-1}}{\lambda^2} & 0\\
  0 & - 1
  \end{bmatrix},\; \tau \geq 0.
\end{equation}

% &\text{maximize} \quad & & \int_{\mathcal{X}}V^{adp}(\boldsymbol{z})c(\mathrm{d}\boldsymbol{z})\\
% &\text{subject to} \quad & & V^{adp}_{i-1}(\boldsymbol{z}) \leq \min_{\boldsymbol{u}\in \mathcal{U}(\boldsymbol{z})}\left\{l(\boldsymbol{z}, \boldsymbol{u}) + \gamma V^{adp}_{i}(\boldsymbol{A}\boldsymbol{z} + \boldsymbol{B}\boldsymbol{u})\right\}\label{eq:adp:iter_ineq_constr}\\
% &&& \forall \boldsymbol{z} \in \mathbb{R}^{6}\times \{-1, 0, 1\},\quad i=1,\dots,M,\\
% &&& V^{adp}_0=V^{adp}_M=V^{adp},

Finally, from~\eqref{eq:sdp:exp_value},~\eqref{eq:sdp:cond1} and~\eqref{eq:sdp:cond2} we can rewrite~\eqref{eq:sdp_upper_bound} as a Semidefinite Program (SDP)~\citep{Vandenberghe:1996fy}:
% \begin{equation*}
 \begin{alignat*}{2}
&\text{minimize} \quad & & \tr\left(\var \mat{P}\right) + r  \\
           & \text{subject to:} %           & &  P \in \mathbb{S}^{n_{\randvar}},\; q \in \mathbb{R}^{n_{\randvar}},\; r \in \mathbb{R}, \; \tau \in \mathbb{R}\\
           && \begin{bmatrix}
    \mat{P} & \vect{q}\\
    \vect{q}\tpose & r - 1
    \end{bmatrix} \succeq \tau  \begin{bmatrix}
  \frac{\var^{-1}}{\lambda^2} & 0\\
  0 & - 1
  \end{bmatrix}
\\
&&& \begin{bmatrix}
    \mat{P} & \vect{q}\\
    \vect{q}\tpose & r
    \end{bmatrix} \succeq 0, \quad \tau \geq 0.
\end{alignat*}
% \end{equation*}
Since the ellipsoid is centered at the origin, it is possible to choose $\vect{q}=0$ and rewrite the problem as:
% \begin{equation*}
 \begin{alignat}{2}\label{eq:prob_boyd_cheb}
&\text{minimize} \quad & & \tr\left(\var \mat{P}\right) + r  \\
           & \text{subject to:} %           & &  P \in \mathbb{S}^{n_{\randvar}},\;  r \in \mathbb{R}, \; \tau \in \mathbb{R}\\
           && \mat{P} \succeq \tau  \frac{\var^{-1}}{\lambda^2}, \quad r \geq 1 - \tau\\
    &&& \mat{P} \succeq 0, \; r \geq 0,\; \tau \geq 0.
\end{alignat}
% \end{equation*}
The objective function and the constraints in~\eqref{eq:prob_boyd_cheb} are linear in the optimization variables. The optimal solution therefore exists at the boundary of the feasible region: i.e. $r = 1 - \tau$. We need to distinguish two different cases. The first corresponds to
\begin{equation*}
r = 0 \iff \tau = 1 \implies \mat{P} = \frac{\var^{-1}}{\lambda^2},
\end{equation*}
with the optimum being $n_{\randvar} / \lambda^2$. The second case is
\begin{equation*}
\tau = 0 \iff r = 1 \implies \mat{P} = 0,
\end{equation*}
and the corresponding optimum is $1$. Finally, by computing the inverse coordinate transformation to get back $\randvar$, it is possible to write the multivariate inequality explicitly as
\begin{equation}
\Prob \left((\randvar - \mean)\tpose \var^{-1}(\randvar-\mean) \geq \lambda^2\right) \leq \min \left\{1,\frac{n_{\randvar}}{\lambda^2}\right\}.\tag*{\qed}
\end{equation}

\section{Applications}
\label{sec:Applications}
Applications of this result include any field wherein the Chebyshev inequality must be applied to distributions for which the mean and covariance are unknown.

A direct application of this empirical Chebyshev inequality is outlier detection. Given a probability bound, it is possible to compute a threshold $\lambda$ and construct a confidence ellipsoidal set from the sample mean and covariance of the first $N$ samples. Then, if the Mahalanobis distance of the $N+1$\textsuperscript{th} sample exceeds $\lambda$, it can be considered an outlier. In \citep{Hardin:2005bz} a similar approach is described making use of the quantiles of the chi-square- or F-distributions in case of normal data. We expect our bound to give more conservative results than the method proposed in \citep{Hardin:2005bz}, but with more general validity since we make no assumptions on the samples' distribution.

Another application involves solving stochastic optimization problems using data-driven information about the uncertainty without knowing its distribution. Following the approach in~\citep{Chen:2007hoa} and~\citep{Bertsimas:2013wga}, we can make use of our empirical Chebyshev inequality to construct ellipsoidal uncertainty sets with predefined probability guarantees. We can then approximate stochastic programs, that are intractable in their general form~\citep{Shapiro:2005dk}, with robust optimization problems~\citep{BenTal:2009wu} and enforce the optimal solution to be feasible for all the uncertainty realizations inside our ellipsoidal uncertainty set. The latter condition implies the same probabilistic guarantees on the original stochastic program. In certain cases, e.g.\ when the constraints are linear and the uncertainty enters linearly in the coefficients, the robust reformulations are convex and can be solved efficiently as second-order cone programs~(SOCPs)~\citep{BenTal:2009wu}.

%CONCLUSIONS
%---------------
\section{Conclusions}
We have derived a generalization of the empirical Chebyshev inequality in multiple dimensions with the only requirement that the given samples are independent and identically distributed. The derived bound scales linearly with the dimension of the random vector and has the same structure as the one-dimensional inequality.

Since many of the common distributions studied in both theory and practice are unimodal, an interesting improvement of this result could be to introduce the assumption of unimodality in order to derive less pessimistic bounds. Another possible extension is to investigate other norms (e.g. $\infty$ or $1$-norm) and compare the right-hand side of the respective reformulations to understand which one is more appropriate for different kinds of distributions.

There are many possible application of this theoretical result appearing whenever the Chebyshev inequality is employed without knowing the population distribution. In particular, the inequality can be exploited to construct confidence sets which can be used in several situations such as outliers detection or stochastic programs reformulations.

%----------------
%BIBLIOGRAPHY
\bibliographystyle{chicago}
\bibliography{bibliography}

\end{document}